\documentclass{article}
\usepackage{arxiv_adj}

\usepackage{hyperref}
\hypersetup{colorlinks=true,linkcolor=blue,citecolor=blue,filecolor=magenta,urlcolor=cyan}
\usepackage{amsmath,latexsym,amssymb,graphicx,dsfont,amsthm,amsfonts}
\usepackage[linesnumbered,ruled,vlined]{algorithm2e}
\SetKwInput{KwInput}{Input}        
\SetKwInput{KwOutput}{Output} 

\newtheorem{lemma}{Lemma}

\begin{document}


\title{Parasite Chain Detection in the IOTA Protocol}

\author{Andreas Penzkofer, Bartosz Kusmierz$^{1}$, Angelo Capossele, William Sanders, Olivia Saa$^{2}$ \\
\small{IOTA Foundation, 10405 Berlin, Germany}\\
\small{$^1$ Department of Theoretical Physics, Wroclaw University of Science and Technology, Poland}\\
\small{$^2$ Department of Applied Mathematics, Institute of Mathematics and Statistics, University of São Paulo, São Paulo, Brazil
}}

\maketitle

\begin{abstract}
In recent years several distributed ledger technologies based on directed acyclic graphs (DAGs) have appeared on the market. Similar to blockchain technologies, DAG-based systems aim to build an immutable ledger and are faced with security concerns regarding the irreversibility of the ledger state. However, due to their more complex nature and recent popularity, the study of adversarial actions has received little attention so far. In this paper we are concerned with a particular type of attack on the IOTA cryptocurrency, more specifically a \textit{Parasite Chain} attack that attempts to revert the history stored in the DAG structure, also called the \textit{Tangle}.

In order to improve the security of the Tangle, we present a detection mechanism for this type of attack. In this mechanism, we embrace the complexity of the DAG structure by sampling certain aspects of it, more particularly the distribution of the number of approvers.
We initially describe models that predict the distribution that should be expected for a Tangle without any malicious actors. We then introduce metrics that compare this reference distribution with the measured distribution. Upon detection, measures can then be taken to render the attack unsuccessful.
We show that due to a form of the Parasite Chain that is different from the main Tangle it is possible to detect certain types of malicious chains. We also show that although the attacker may change the structure of the Parasite Chain to avoid detection, this is done so at a significant cost since the attack is rendered less efficient.
\end{abstract}

\section{Introduction}

With the arrival of Bitcoin \cite{bit1} a new decentralized payment system based on a trust-less peer-to-peer network has established. Bitcoin is essentially a protocol for reaching consensus between independent entities - that do not need to trust each other - on a chronologically ordered record of transactions. This data structure, which is also called a {\it{blockchain}}, is now a cornerstone of many other Distributed Ledger Technologies (\textbf{DLT}s) \cite{ethereum}, \cite{algorand}, \cite{cardano}, \cite{cryptocurrencySurvey}.

Despite their great success, events of congestion and the correlated high transaction (\textbf{tx}) fees \cite{Huberman_monopolist} show that limitations in throughput, which can become costly, exist for these types of DLTs. Since these events are effectively created due to the bottleneck of a limit of txs that can be processed, scalability has become a core research topic. Furthermore, this issue also hinders the adoption of the technology for applications such as Internet of Things (IoT) \cite{dorri2017}. 

To overcome scalability issues, several techniques have been proposed ranging from increasing block size and frequency, to side chains \cite{croman2016}, ``layer-two'' structures like Lightning Network \cite{light}, Sharding \cite{luu2016}, and other consensus mechanisms \cite{algorand,cardano,casper}. Some DLT's have also replaced the blockchain structure with a directed acyclic graph (\textbf{DAG}). This approach is used in IOTA \cite{I_WP} and other protocols \cite{Lewenberg2015}, \cite{byteball}, \cite{Boyen2016}, \cite{spectre}, \cite{meshcash}, \cite{conflux}, \cite{dexon}, \cite{vegvisir}. DAG-based protocols reach consensus on a {\it{partially}} ordered log of transactions, allowing the log to have width and which can increase the throughput of the system.

\subsection{The IOTA protocol}

In this paper, we study the DAG-based IOTA protocol introduced in \cite{I_WP}, where transactions are recorded in a DAG dubbed the {\it Tangle}. The vertices in this DAG are transactions. If there is an edge between two transactions $x\leftarrow y$, we say that $y$ (directly) approves $x$. If there is a directed path from $y$ to $x$ but no edge, $y$ indirectly approves $x$. A transaction with no approvers is called a {\it tip}. Under the IOTA protocol, all incoming transactions attach themselves to the Tangle by approving two (not necessarily distinct) tips.  

Since the DAG structure is heavily determined by the order and manner in which txs are approved, the algorithm used to select the tips which a new transaction will attach itself to, is of critical importance. It is important to note, that the nodes in this protocol are free to choose from the array of available tip selection methods. In this paper we focus on two tip selection algorithms:

Firstly, the {\it Uniform Random Tip Selection} (\textbf{URTS}), is a very basic algorithm: we simply select a tip from the set of all available tips with a uniform random distribution. Despite being the most efficient method numerically, it would also be accompanied by security vulnerabilities and allow for tip selection behavior that is non-beneficial for the safety of the Tangle \cite{I_WP}. 
However, as we will see in Section \ref{sec:URW} it is closely related to the next algorithm, which is the one that most resembles the current implementation in the protocol. More particularly, the analytical derivations for the following algorithm depend on the solutions for this URTS algorithm.

The second tip selection algorithm employs a Monte Carlo Markov Chain: here we select the tip at the end of a random walk (\textbf{RW}), beginning at the first tx in the Tangle. In the current implementation of the IOTA protocol, the RW can be biased towards transactions with large cumulative weight, which is a tx's own weight plus the sum of all own weights of directly or indirectly approving txs. The amount of bias is determined by a parameter $\alpha$. When $\alpha = 0$, we dub the RW as {\it Unbiased Random Walk} (\textbf{URW}). Consequently, when $\alpha > 0$, we dub it {\it Biased Random Walk} (\textbf{BRW}). As we will see in Section \ref{sec:BRW} under the current network conditions and implementation of the IOTA protocol, the BRW has very similar properties to the URW and, therefore, it is sufficient to study the URW.

\subsection{The Parasite Chain attack}

Due to the probabilistic nature of the immutability of the ledger state, cryptocurrencies are subject to certain security concerns \cite{Li_securityblockchain}. One such concern is that an adversary may revert the ledger to an earlier state if he possesses a sufficient amount of hashing or voting power \cite{Rosenfeld2014}. In practice, such an adversarial action would result in a fork of the ledger and enable the possibility for a double-spend of funds.

\cite{I_WP} describes several attack scenarios under which an adversary may attempt a double-spend on the IOTA protocol. In this paper we focus on the Parasite Chain (\textbf{PC}) attack. 
In this attack the adversary places a value tx in the main Tangle, whilst also creating a side chain in secret that contains a double-spending tx, see  Fig. \ref{fig:PC_overview}. 
Once the PC is revealed to the network, the attacker then exploits the tip selection algorithm by leading honest txs to approve the PC instead of the main Tangle. This is possible, through the following mechanism: firstly, by attaching the PC to a particular tx (the \textit{root} tx), the cumulative weight of this root tx can be significantly increased and the RW tip selection algorithm will then be drawn towards this root. In addition, the attacker can also increase the number of links from the PC to this root tx, to increase the probability for the RW to continue onto the PC. If the attack is successful, the part of the Tangle that approves the originally visible double-spend tx is abandoned and the PC becomes the new main Tangle, thereby changing the ledger history.

\cite{cullen_PC} discuss this attack in further detail and show that for certain values of the parameter $\alpha$ in the BRW, the attack has an increased likelihood to succeed. Furthermore, although the security is improved with increasing value of $\alpha$, it should not be selected too high, otherwise txs would be left behind and excluded from the ledger. Therefore, $\alpha$ has to be selected in a manner that is a compromise between the safety of the system and avoiding orphanage of txs, i.e. txs not being approved. More generally \cite{cullen_PC} also showed that the success of the attack also depends on other variables, such as the time that the PC is revealed to the main tangle and the number of root txs the PC is attached to the main Tangle.

\begin{figure}[t]
\centering
\includegraphics[width=0.5\textwidth]{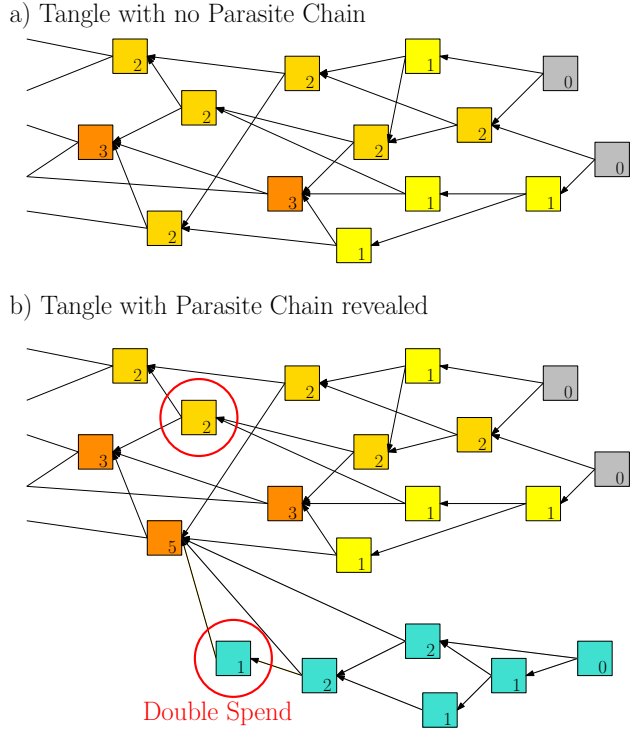}
\caption{Examples of parts of a Tangle. Transactions  are  represented  by  squares and  approvals  by  arrows. The number of approvers of each transaction is in the bottom right corner of each transaction. Parts (a) and (b) represent the Tangle with and without a conflicting Parasite Chain attached. Transactions that conflict with the most recent part of the main Tangle are shown in blue.}
\label{fig:PC_overview}
\end{figure}

In this paper, we present a method to counter the vulnerability of the PC attack by introducing certain detection mechanisms. These detection methods employ the knowledge of the underlying structure of the Tangle, more particularly the likelihood for txs to have a certain number of direct approvers. We show that by measuring the distribution of approvers for a selected set and comparing it to a reference distribution, the Tangle can be checked for abnormalities. Upon flagging a suspicious part of the Tangle as a PC, countermeasures can be taken, such as restarting the RW with an increased $\alpha$ value \cite{ferraro}. Since the tip selection method can be switched or restarted immediately during the detection, an honest tx issuer who detected the PC, would avoid approving tips of the PC and hence contribute towards rendering the attack unsuccessful. We show that an adversary would have to significantly reduce the efficacy of the attack if he wants for the PC to remain undetected.

\subsection{Contributions of this paper}

Our contributions in this paper are twofold:
firstly, we present analytical models that capture the underlying structure of the Tangle, in the form of the likelihood for txs to have a certain number of direct approvers. We present models for the URTS and the URW, and we show that the distribution for the URW is tightly linked with the distribution for URTS. We compare the analytical models to simulation results and show that generally a good agreement is reached in the high load regime. In the low load regimes, the model predicts the values well only if txs are only allowed to approve the same tx once.

Secondly, we describe a method of how to reduce the susceptibility of the IOTA cryptocurrency towards a specific type of double-spend attack, more particularly a Parasite Chain. This enables a proactive tool against malicious actors and improvements for security. We measure how 'distant' a certain sample set of txs is from the derived distributions and we show that we can employ the distance metric, to effectively detect simple versions of a PC. If the attacker decides to avoid the detection methods, he is forced to build the PC in a more complicated way, which is correlated to a decrease of the attack's efficacy. We demonstrate this for two ways of selecting sample sets: through RWs and through collecting the txs that directly or indirectly approve a particular tx, i.e. the future cone of that tx. We conclude that by combining these two methods, a powerful detection tool is provided to honest tx issuers that allows them to make their tip selection more safe and PC attacks less likely to succeed.


\section{Model for the Number of Approvers}\label{sec:model number approvers}

Due to the complexity of the Tangle, it can be difficult to derive exact solutions which describe certain mechanics. On the other hand, deterministic solutions can provide a sufficiently good picture to describe certain mechanics despite their simplified assumptions. Here, we attempt to discuss the likelihood of a randomly selected tx having a given number of direct approvers $n$, through a deterministic model. The model is facilitated by employing a linear approximation of the exit probability distribution, which represents the likelihood of a tip being selected.

Since the tip selection is performed probabilistically, the same tip may get selected twice, although this is only likely to happen in the \textit{low load regime}, i.e. if the rate of arriving txs is low. Since the tip selection algorithm is not enforced, it remains up to the node to decide what to do in this case. Here, we consider two particular scenarios that do not require rerunning the tip selection algorithm and which are, therefore, considered numerically efficient options. Initially, we will discuss both scenarios, before focusing on the former in more detail: in the single edge model (\textbf{SEM}) only one edge is created instead of two, and in the multi-edge model (\textbf{MEM}) two edges are created to the same tx.
As we will show, in the \textit{high load regime}, i.e. when the rate of arriving txs is high, SEM and MEM converge to the same distributions. Since the Tangle is built to allow for high throughput, we will mainly focus on the high load regime and the difference between SEM and MEM can be neglected. However, since the protocol should also be analyzed in the low load scenario and the node is free to choose the tip selection algorithm, the two methods are discussed for completeness and comparison.
\\

Once a tx is selected for the tip selection, it is assumed that the Proof-of-Work plus the propagation to the network equal to the time $h$, and that no other node will be aware of the approval before this delay has passed \cite{I_WP}. Let $p_i$ be the probability of being selected by the tip selection algorithm for the $i$-th tx, $N_i$ the final number of approvers of $i$ and $t_i$ the time of first approval. We notice that $p_i$ depends on the number of tips and hence may change with the arrival of a new tx.
Note that if txs $i_1$ and $i_2$ are the approvees of $i$ then no further txs are attached to them later than $t_i+h$, due to the delay $h$.
It is noteworthy that the following applies for URW: after $t_i+h$ none of the directly or indirectly approved txs, i.e. the entire past cone of $i$, receive any further approvers. The exit probability of the URW at $i$ remains then unaffected after $t_i+h$ by the arrival of new tips and $p_i$, and also the probability for the random walk to pass through the tx, remains constant after that time. This is not the case if backtracking is allowed, i.e. the RW is allowed to return to the past cone of $i$ once it left it, however, this is currently not implemented in the IOTA protocol. 

Throughout this work we will frequently employ the Poisson distribution function
\begin{equation}\label{eq:Poisson}
	P(\gamma,n)=e^{-\gamma}\frac{\gamma^{n}}{n!}
\end{equation}
where $\gamma$ is a rate. 

In the following sections we employ 
\begin{lemma}
The number of approvers for a given tx $i$ is given by
\begin{equation}\label{eq:model_Ni}
	N_i=1+p_{i0}+\text{Pois}(\lambda_i)
\end{equation}
where $\rm{Pois}(\cdot)$ is the Poisson distribution,
\begin{equation}\label{eq:model_lambdai}
	\lambda_i=\lambda\int_{t_i}^{t_i+h}dt \left\{ \begin{matrix} 2p_i(t) & \text{for MEM} \\ 2p_i(t)-p_i(t)^2 & \text{for SEM}\end{matrix} \right.
\end{equation}
$\lambda$ is the tx rate in units of $h$, and
\begin{equation} \label{eq:model_p0}
	p_{i0}=\left\{ \begin{matrix} p_i(t_i) & \text{for MEM, approximately} \\ 0 & \text{for SEM\;\;\;\;\;\;\;\;\;\;\;\;\;\;\;\;\;\;\;\;\;\;\;\;}\end{matrix} \right.
\end{equation}
\end{lemma}
\begin{proof}
Assume there is a Poisson process of arriving txs \cite{I_WP} with rate $\lambda$, i.e. the number of arrivals within the interval $h$ after $i$ receives its first approval is given by the random variable $N$. Let us consider SEM first. The probability of $i$ receiving an additional approver, once a tx arrives is $p_i^+(t)=1-(1-p_i(t))^2$. Hence from the viewpoint of $i$, txs arrive at rate $p_i^+\lambda$ (of which all would reference $i$). Furthermore, for independent events it holds that $\mathbb{P}(N=n)=\mathbb{P}(M_1+M_2=n)$, where $\mathbb{P}(X=x)$ is the probability that the random variable $X$ takes the value $x$, and $M_1$ and $M_2$ are again Poisson processes, with rates $\mu_1$ and $\mu_2$. We can, therefore, assume that the arrivals of attachments to $i$ occur through a series of Poisson processes with rate $p_i^+\lambda dt$ at time intervals $dt$, which leads to the integral form. In the case of MEM, we consider two rates of arriving tx that approve $i$: txs that approve $i$ twice or once, and their rates of arrival are $p_i(t)^2\lambda$ and $(p_i^+(t)-p_i(t)^2)\lambda$, respectively. With the same argument as above, we can find the integral form
\[	N_i=1+\text{Pois}(\lambda_{i1})+2\text{Pois}(\lambda_{i2}) \]
where
\begin{align*}
	\lambda_{i1}&=\lambda\int_{t_i}^{t_i+h} 2p_i(t)(1-p_i(t))  dt\\
	\lambda_{i2}&=\lambda\int_{t_i}^{t_i+h} p_i^2(t)  dt
\end{align*}
These two Poisson processes are independent and can be further combined. Finally, the first approver also has the likelihood $p_{i0}$ to approve the same tip twice. This assumption does not take into account that for small $\lambda$ the Poisson process has a significant impact on the number of approvers. More specifically, for small $\lambda$ the likelihood is increased that a single tip is simultaneously selected by one or multiple txs and hence the probability to have an even number of approving edges is increased compared to an odd number. 
\end{proof}

Generally the average probability to have $n$ approvers is given by 
\begin{equation}\label{eq:P_mean}
	P(n)=\sum_{i=1}^N \mathbb{P}(N_i=n)
\end{equation}
where $N$ is the cardinality of a list of txs that are considered. Since as previously discussed for URW $p_i(t)$ is only variable for a short time (maximally $h$ after $i$ is revealed), we assume $p_i(t)=p_i \forall t$ is approximately true. The integral forms in (\ref{eq:model_lambdai}) can, therefore, be further simplified, which we employ in the next section.

\subsection{Uniform Random Tip Selection}\label{sec:URTS}

In the URTS algorithm, tips are selected at random from the set of $L$ available tips and with the most recent assumption $p_i=L^{-1}$ on average. According to \cite{I_WP}, the number of tips is approximately $L=2\lambda$. However, for small $\lambda$, the number of tips is limited by one instead. For simplicity, we assume that
\begin{equation}\label{eq:model_tips}
	L=1+2\lambda
\end{equation} 
Using (\ref{eq:P_mean}) and (\ref{eq:model_Ni}) the probability to have $n$ approvers (or equivalently $n-1$ beyond the first) is then given by the Poisson distribution function 
\begin{equation}
    P_{U}(n)=P(\lambda_U,n-1)
\end{equation}
with the rate
\begin{equation*}
	\lambda_{U}= \left\{ \begin{matrix} 2\lambda L^{-1} & \text{for MEM} \\ 2\lambda L^{-1}(1-0.5L^{-1}) & \text{for SEM}\end{matrix} \right.
\end{equation*}
Note that in the high load regime (i.e. $\lambda$ is large), the quadratic term can be neglected and MEM and SEM converge to the same rate, as previously discussed.

\begin{figure}[ht]
\begin{minipage}{0.45\textwidth}
    \hspace{-1.0cm}
    \includegraphics[width=1.3\textwidth,trim=0 0 0 0,clip=true]{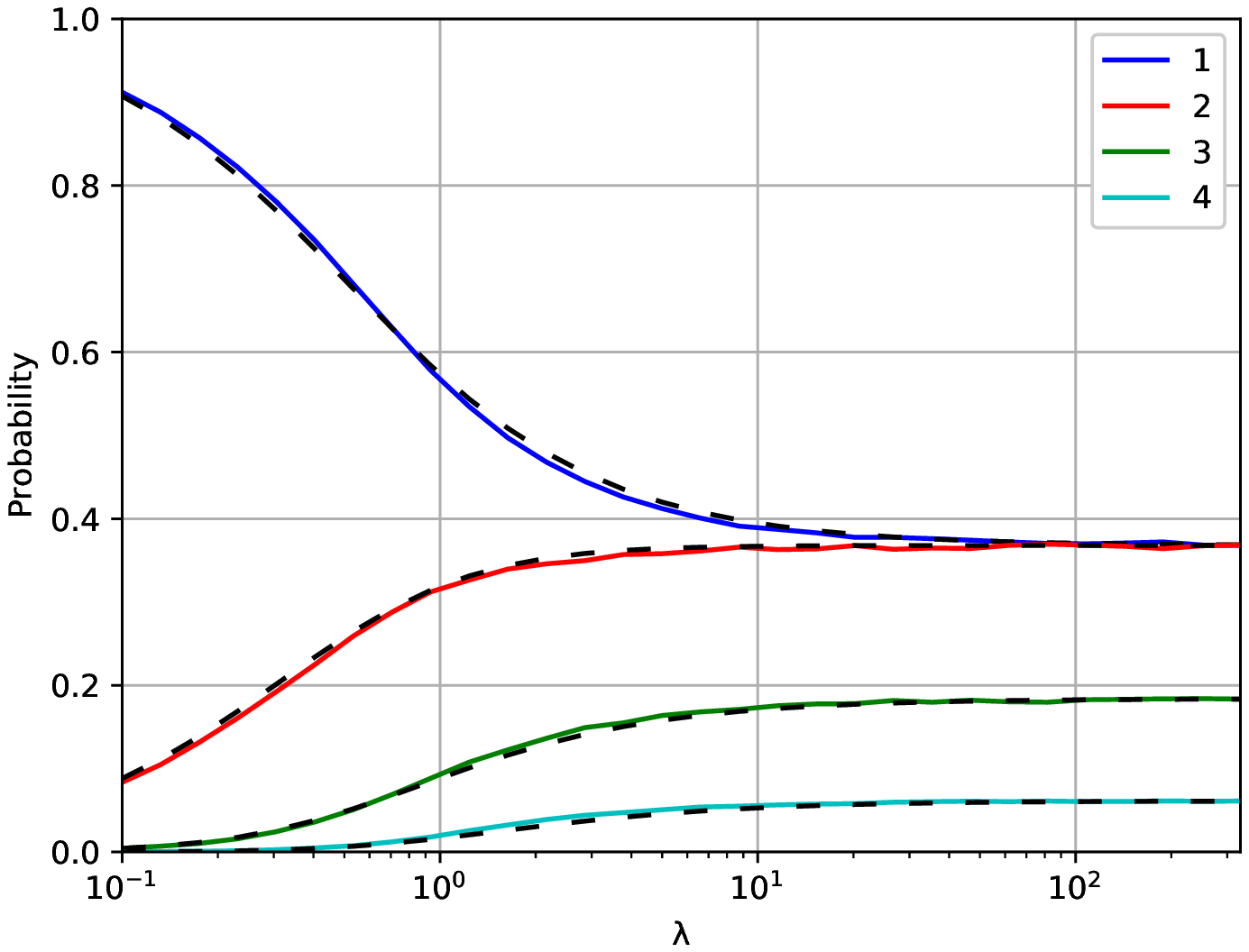}
    a) Single-Edge model (SEM)\newline
\end{minipage}\hfill
\begin{minipage}{.45\textwidth}
    \hspace{-0.3cm}
    \includegraphics[width=1.3\textwidth,trim=0 0 0 0,clip=true]{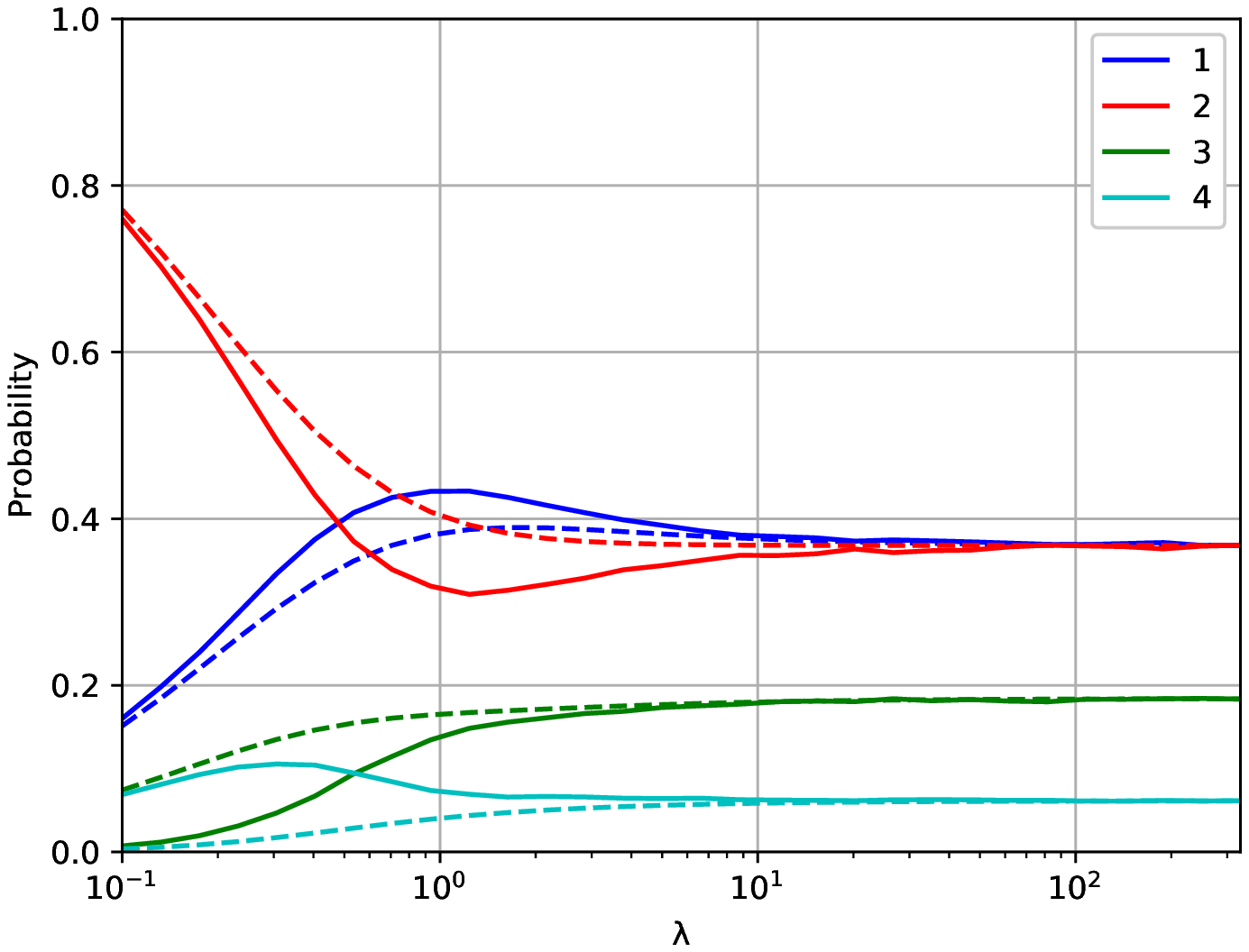}
    b) Multi-Edge model (MEM)\newline
\end{minipage}\hfill
\vspace{-.2cm}
\caption{Probability for a randomly selected tx, in a Tangle created by URTS, to have $n$ approvers with $\lambda$ for $n=\{1,..,4\}$. Values from simulation results and analytically predicted values are shown with continuous and dashed lines, respectively.}
\label{fig:URTS_SEM_MEM}
\end{figure}

Fig. \ref{fig:URTS_SEM_MEM} shows a comparison of the numerical and analytical values of $P_{U}(n)$ for both SEM and MEM. It can be seen that for SEM the model represents the numerical values well. In both models, the distribution converges towards the same distribution in the high-load regime. This is because for large $\lambda$ almost none of the tips are selected twice. However, in the low load regime (i.e. $\lambda < 10$) the predicted values do not match for MEM. This is due to the discrete nature of the Poisson process, which as described above, is not accounted for. Furthermore, even numbers of approvers have a higher probability (see $n=4$), while odd numbers of approvers are less likely than predicted (see $n=3$).

\subsection{Unbiased Random Walk }\label{sec:URW}

The tip selection algorithm that is currently employed in IOTA is the Monte-Carlo-Markov-Chain RW with the parameter $\alpha$ \cite{I_WP}. Typically, analytical solutions are easier to find by initially considering $\alpha=0$, where effects such as txs that are left behind, or the dependence of the RW on the cumulative weight, play no role. Here we follow the same path, before increasing the complexity of the analysis by considering $\alpha>0$ in Section \ref{sec:BRW}. We employ SEM for the analysis in this section, since the analytical model provides more accurate results in this case. However, for most of our analysis, both models would be appropriate, since both predict the numerical values well in the high-load regime, which is in the focus in this work. Note that \cite{Kusmierz_continsim} observed that for $\alpha=0$, the mean tip number converges to $L\approx 2.1\cdot \lambda$. However, employing this observation does not lead to improvements for this model and hence the same average tip number as for URTS is employed.

For a given number of tips $L$, we can order the exit probabilities by their likelihood. We then define the $L$-normalized exit probability $e(x)$ to be the exit probability distribution that is normalized from the index interval $\{1,..,L\}$ onto the relative index interval $(0,1]$. For large enough $\lambda$, the expected exit probability of the $i-th$ most likely tip is then approximately given by
\begin{equation}
    \mathbb{E}e_L(i)\approx\int_{(i-1)/L}^{i/L}e(x)dx
\end{equation}
The $L$-normalized exit probability can be expressed as 
\begin{equation}\label{eq:palpha0}
    e(x)=1+f(x)
\end{equation} 
where $x \in [0,1]$, $\int^1_0 f(x)=0 $ and $f(x)>-1$.

Fig. \ref{fig:exitprobability} shows the numerically calculated $L$-normalized exit probability for URTS and URW, as well as a linear fit for URW. $e(x)$ has numerically been calculated by averaging over $10^3$ samples of tip sets. For each tip set, $10^6$ tip selections are performed and the tips are ordered by their exit probability $e_L(i)$. Despite the high number of tip selection samples, the values of indices close to 0 (1) are slightly over- (under-) estimated due to stochastic effects, as can be seen for URTS where the value is expected to be $e_{U}(x)=1$ over the entire interval. It can be seen that for the fitted curve the numerical values agree well for most of the range, apart from the smaller indices, where the exit probability is noticeably increased, suggesting that certain tips in the Tangle have a considerably higher probability to be selected.

\begin{figure}[ht]
\centering
\includegraphics[width=0.6\textwidth,trim=0 0 0 0,clip=true]{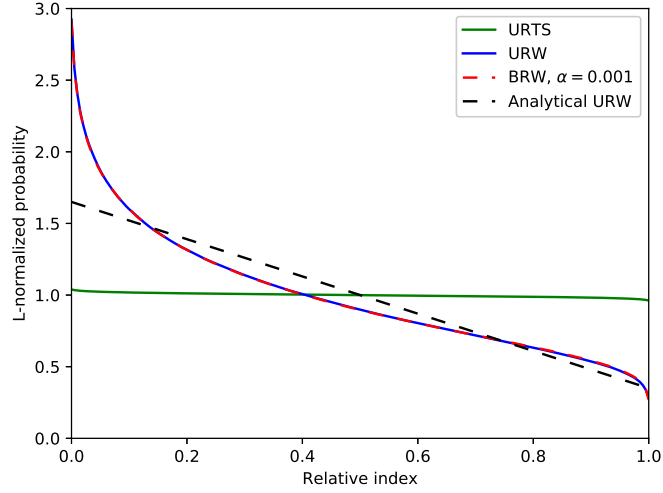}
\caption{$L$-normalized exit probability with the ordered relative index number for $\lambda=100$. The curve for the analytical linear approach is shown for $a=1.3$.}
\label{fig:exitprobability}
\end{figure}

For a given relative index number $x$, we can assume a Poisson process of txs and the probability density of having $n$ approvers is then given by
\begin{equation}\label{eq:p_urw}
    p_{U\!RW}(n,x)=P(e_{U\!RW}(x)\lambda_U,n-1)
\end{equation}
The expected value is then given by
\begin{align}
    P_{U\!RW}(n) &=\int^1_0 dx \;e^{-e_{U\!RW}(x)\lambda_U}\frac{(e_{U\!RW}(x)\lambda_U)^{n-1}}{(n-1)!}\\
     & = P_{U}(n)\int^1_0dx\;e^{-f(x)\lambda_U}(1+f(x))^{n-1} \nonumber
\end{align}
As can be seen, the distribution for the URW can be decomposed into a product of the distribution for URTS times a factor dependent on the exit probability function. This highlights how these two tip selection algorithms are correlated, and why their distribution appears similar.

For URW, the probability to pass through a given tx, is equal to the final value of the exit probability during the time the tx was a tip. Hence to derive the approver statistics observed along the path of RWs, (\ref{eq:p_urw}) is weight by $e(x)$, and the expected distribution is
\begin{align}
    P^*_{U\!RW}(n) &=\frac{1}{b}\int^1_0 dx \; p_{U\!RW}(n,x)e(x)
\end{align}
where 
\begin{equation*}
    b = \sum_{n=0}^\infty \int_0^1 dx \; p_{U\!RW}(n,x)e(x)
\end{equation*}
normalizes the probability.
\\

\begin{figure}[t]
\centering
\includegraphics[width=0.7\textwidth,trim=0 0 0cm 0cm,clip=true]{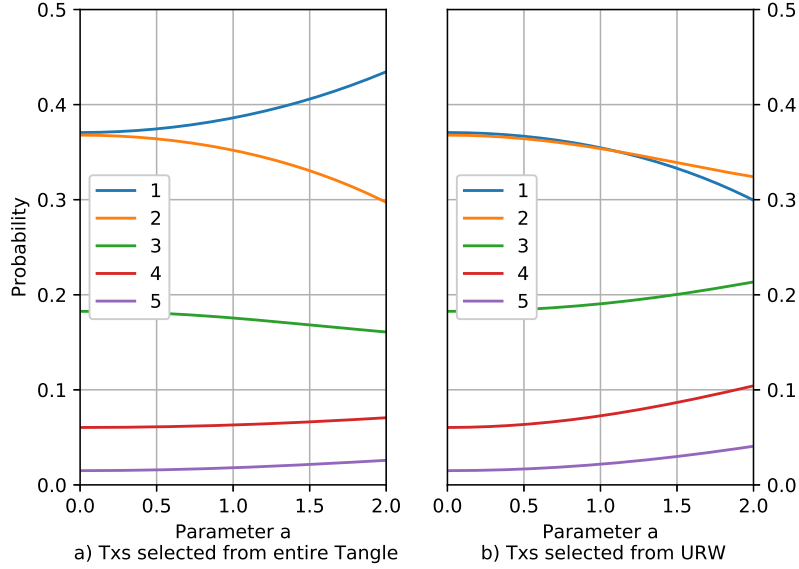}
\caption{Values of (\ref{eq:PURW}) and (\ref{eq:PURWstar}) with the parameter $a$ for several values of approver numbers $n$. $\lambda=100$}
\label{fig:parametera}
\end{figure}

\subsubsection*{Linear approach}\label{sec:model_alph0_simple}
A simplified approach is presented by employing the distribution $f(x)=a(x-0.5)$, $a\in [0,2]$, where $a$ is limited to achieve only positive values for $f$. The expected value is then given by 
\begin{align}
	P_{U\!RW}(n) 	
	&=P_{U}(n)g(n-1) \label{eq:PURW}\\ P^*_{U\!RW}(n)&=\frac{P_U(n)g(n)}{\sum_{m=1}^\infty P_U(m)g(m)} \label{eq:PURWstar}
\end{align}
where
\begin{align*}
    g(n)&=\frac{1}{a} \sum_{j=0}^{n}\lambda_U^{-j-1}\frac{n!}{(n-j)!}
	\left[ e^{-y\lambda_U }(1+y)^{n-j} \right]^{-0.5a}_{0.5a}
\end{align*}
Note that for $a\to 0$ (\ref{eq:PURW}) converges towards URTS: 
\begin{align*}
	\frac{P_{U\!RW}(n)  }{P_U(n)}
	&= -\sum_{j=0}^{n-1}\frac{(n-1)!}{(n-1-j)!} (n-2-j)=1
\end{align*}

Fig. \ref{fig:parametera} shows the probability to have $n$ approvers with the value of the parameter $a$. As can be seen in Fig. \ref{fig:parametera}a), introducing a gradient in the exit probability shifts the likelihood of having a certain number of approvers away from the mean towards less (one approver) or higher numbers of approvers (more than 3), corresponding to increased or decreased exit probabilities. On the other hand, it can be seen from Fig. \ref{fig:parametera}b) that the likelihood for a URW to visit sites with a higher number of approvers, is shifted to only txs with higher numbers (more than 2), since these are sites more frequently visited by the URW.

From Figs. \ref{fig:exitprobability} and \ref{fig:alpha0_SEM}, we can see that a good agreement is achieved with the numerical results for $a= 1.3$. 

\begin{figure}[ht]
\begin{minipage}{0.45\textwidth}
    \hspace{-1.0cm}
    \includegraphics[width=1.3\textwidth,trim=0 0 0 0,clip=true]{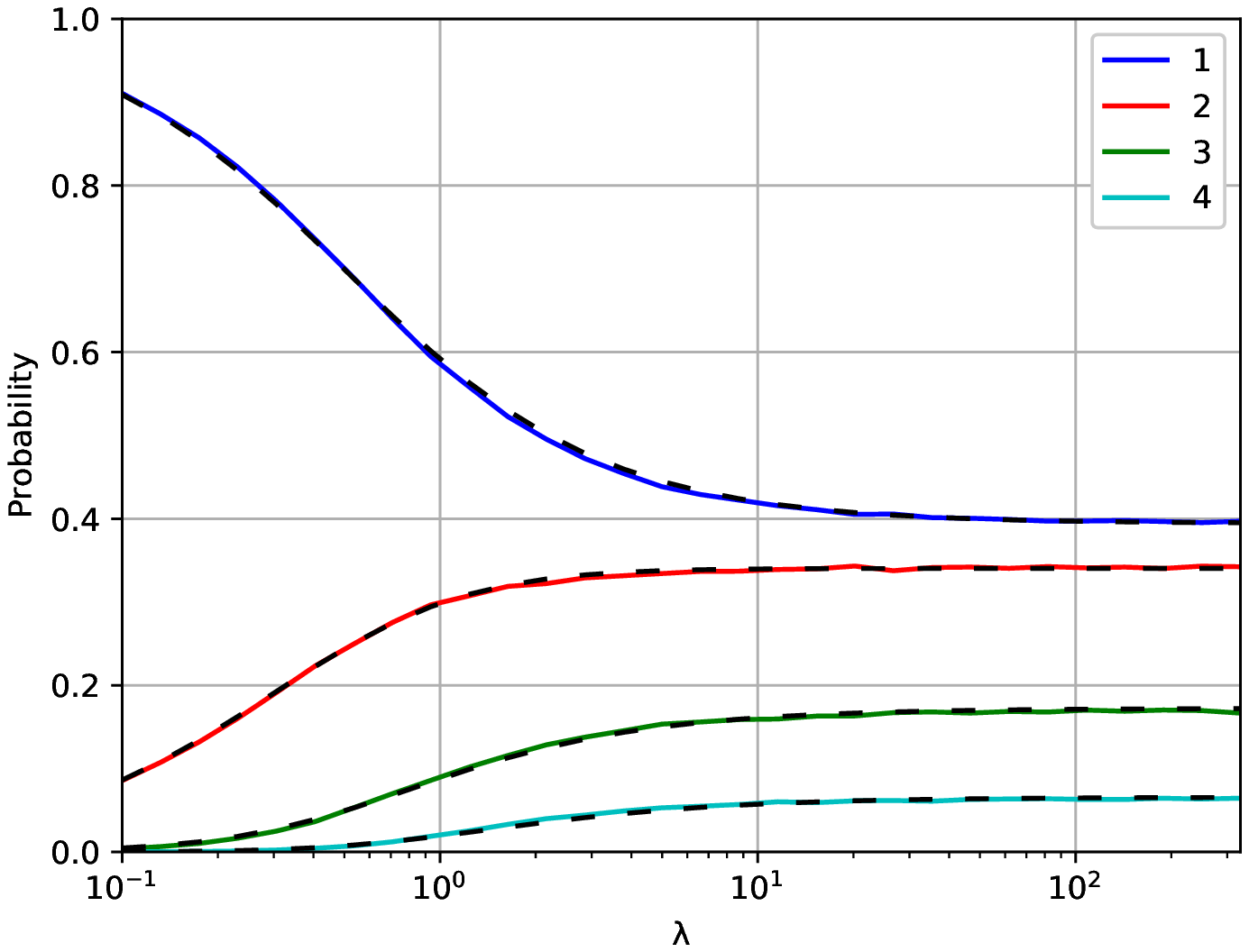}
    a) Randomly selecting a tx from the entire Tangle \newline
\end{minipage}\hfill
\begin{minipage}{.45\textwidth}
    \hspace{-0.3cm}
    \includegraphics[width=1.3\textwidth,trim=0 0 0 0,clip=true]{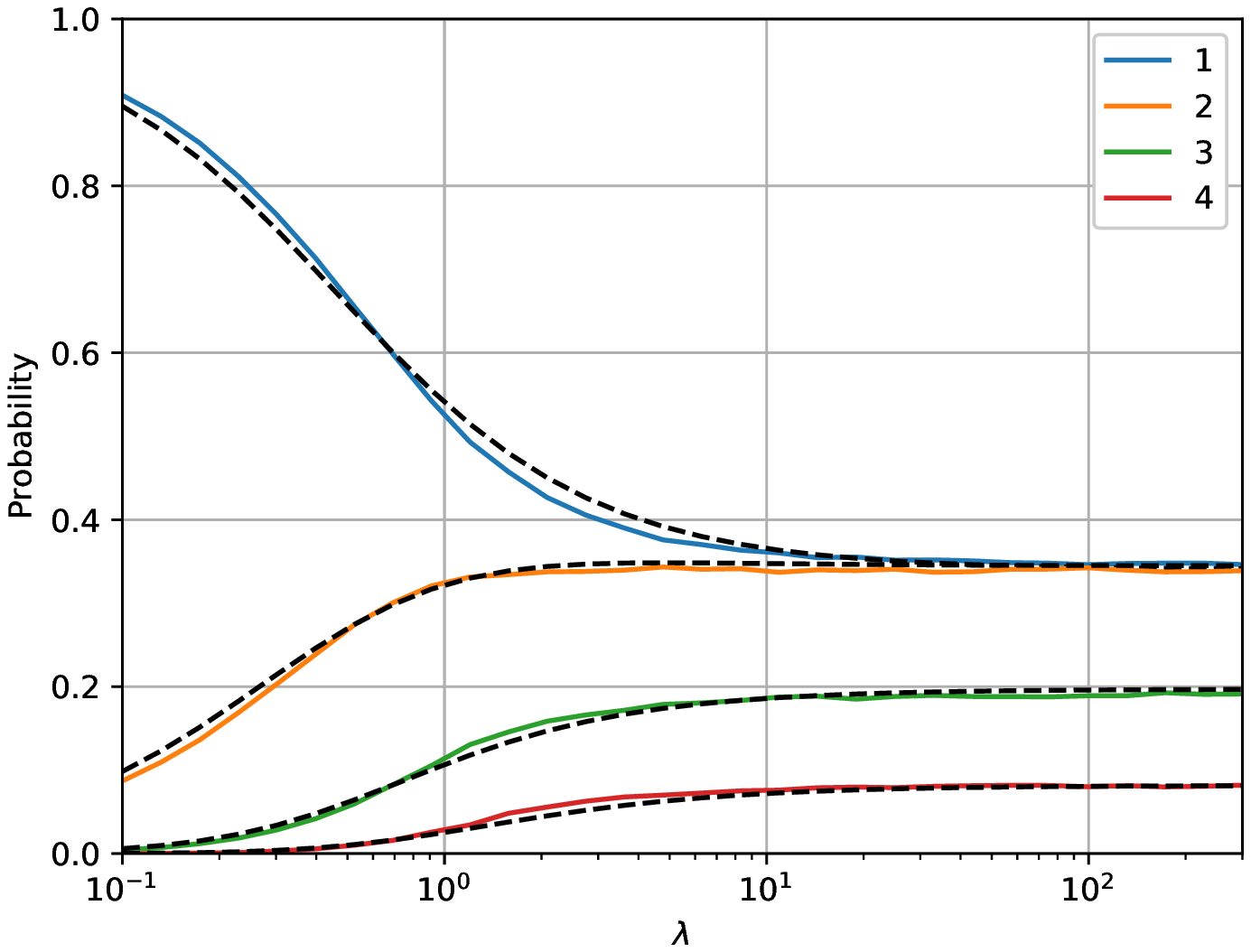}
    b) Randomly selecting a tx from an URW\newline
\end{minipage}\hfill
\vspace{-.2cm}
\caption{Probability for a tx to have $n$ approvers with $\lambda$ for $n=\{1,..,4\}$; for SEM and URW tip selection, and $a=1.3$. Values from simulation results and analytically predicted values are shown with continuous and dashed lines, respectively.}
\label{fig:alpha0_SEM}
\end{figure}


%
\subsection{Biased Random Walk} \label{sec:BRW}

In the IOTA protocol, the RW is made dependent on the cumulative weight of a tx. The cumulative weight $w_x$ of a tx $x$ is defined as the sum of its own weight, plus all own weights of txs directly or indirectly approving the tx. The dependency on the RW is introduced to prevent undesirable behaviors, such as the selection of old tips (lazy tip selection), or certain types of PCs \cite{I_WP}.
In the current implementation of the Tangle, the path of the RW depends, therefore, on a parameter $\alpha$ and the cumulative weight of txs encountered along the path. More specifically, the probability to transition from a tx $x$ to a tx $y$ is given by \cite{I_WP}:
\begin{equation}
	P(x \leadsto y ) = \frac{\exp (\alpha w_y ) }{ \sum_{ z} \exp (\alpha w_z )}
\end{equation}
where the denominator normalizes the probability by summing over all direct approvers $z$ (including $y$) of $x$.

It should be noted that the value of the parameter $\alpha$ ought to be selected carefully. Selecting the value too low would result in the BRW having little or no improvement to the security compared to URW, while a too high value would result in txs being orphaned. Typically, orphanage behavior is observed in simulations at about $\alpha \lambda>1$ \cite{POBLB}. The default value in the IOTA reference implementation of $\alpha=0.001$ \cite{iota.go} ensures that the Tangle structure is the same at least up to a tx rate of $\lambda=100$, since no orphans are created and the exit probability is the same, as can be seen in Fig. \ref{fig:exitprobability}.

With the Minimum Weight Magnitude parameter set to 14 (current setting of the difficulty in the IOTA protocol), the average time for the PoW in a Core i7 platform is about $4.1s$ \cite{Elsts}. At a tx rate in the order of 5 tps (current, over the scope of one day, measured tps by an IRI node \cite{iota.go}), the expected value for the tx rate is $\lambda\approx20$. Since $\lambda\gg 1$, the difference between SEM and MEM is negligible. Furthermore, since $\lambda\alpha\ll1$, the difference between the exit probabilities and hence the approver distributions of URW and BRW are negligible, as discussed in the previous paragraph. We, therefore, employ the models developed in Section \ref{sec:URW} for the following sections.

\section{Parasite Chain detection}\label{sec:PC detection}

The vulnerability of the Tangle to the PC attack is first recognized in \cite{I_WP}. A PC is a specific type of a double-spend attack, where the attacker issues a tx on the main Tangle and, simultaneously, issues a double-spending tx, which is not yet revealed. He then continues to issue further txs that approve the double-spend on the PC in secret, to strengthen the PC chain, till the receiver accepts the main Tangle double-spend tx. Upon receiving the goods or value from the to-be-swindled receiver, the adversary then reveals the PC in the hope that it can outpace the incompatible part of the current main Tangle, by attracting as many tip selection RWs as possible. If successful, the part of the Tangle that approves the main Tangle double-spend tx becomes orphaned, and the PC becomes the new main Tangle, thereby invalidating the main Tangle double-spend tx.

In this paper, we focus on a specific Parasite Chain attack where the adversary pins the PC to one or more ($k$) honest txs in the main Tangle, which are issued at a similar time as the double-spend tx on the main Tangle. We dub this type of a PC $k$-pinned PC. Generally, an adversary may choose to attach to multiple root txs. However, for simplicity we assume a 1-pinned PC.

We assume that the adversary attempts to issue as many txs as possible that directly approve the PC root, as well as the previously issued PC tx. This approach is taken for two reasons: firstly, by directly or indirectly approving this root tx, the cumulative weight is increased both by the txs in the main Tangle, as well as by the PC. This can significantly increase the cumulative weight of the root, compared to the double-spend tx in the main Tangle. As a result, the RWs that are performed for tip selection will be drawn towards the root. Secondly, by attaching as many direct approvers to the root, the probability for the RW to hop onto the PC once the RW passes through the root, depends on the number of links to the PC. It can, therefore, be increased by directly attaching as many malicious txs as possible. Fig. \ref{fig:PC_types}a) shows the most efficient way of attaching a 1-pinned PC, and we dub it a simple PC (\textbf{SPC}).
\\

\begin{figure}[t]
\centering
\vspace{0.8cm}
\includegraphics[width=0.65\textwidth]{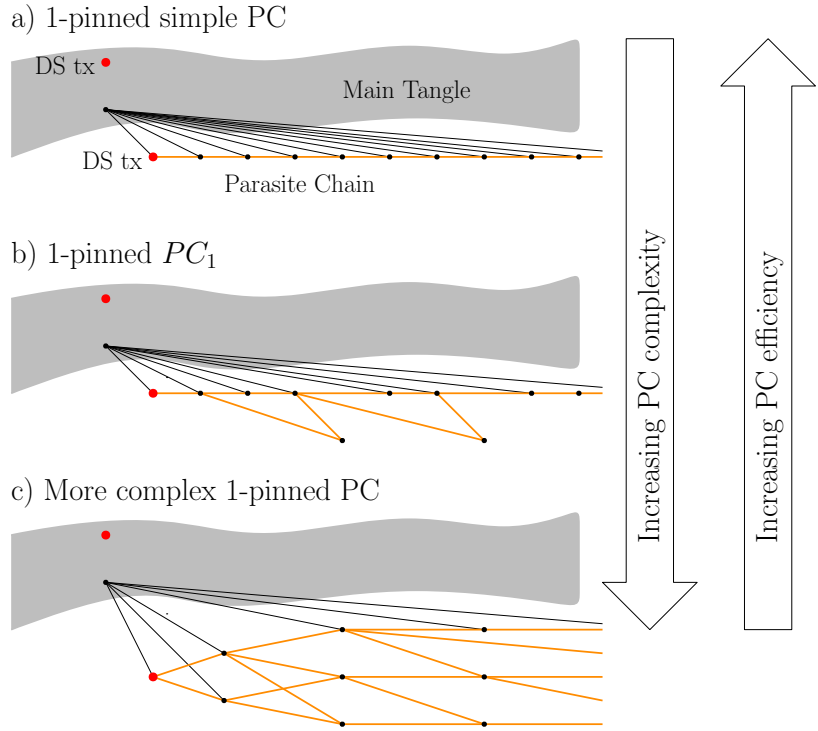}
\caption{Different types of PCs with an increasing effort of the adversary to hide the PC from the detection methods. The number of indicated malicious txs is kept the same for a)-c). With higher complexity, fewer links to the root tx in the main Tangle are created. }
\label{fig:PC_types}
\end{figure}

In this section, we describe a methodology for how we can employ the distributions that we derived in the previous section, to detect a pinned PC. More particularly, since the distribution of approvers in the PC in Fig. \ref{fig:PC_types}a) is significantly different from the distributions derived in the Section \ref{sec:model number approvers}, we propose the implementation of a metric that compares the encountered approver distribution to the expected one, and allows to identify PCs. This would allow nodes to actively counter the creation and success of PCs.

We define the following distance metric \begin{eqnarray}\label{eq:d_P}
    d_{P}=\frac{1}{2}\sum_{n=0}^\infty |P(n,S)-P_{ref}(n)|
\end{eqnarray}
where $P$ is the measured distribution, or probability vector, that is deduced by measuring on the sample of size $S$. In this paper, the sample size is varied between $10$ and $100$, see Fig. \ref{fig:distRWs_CDF}. $P_{ref}$ is the reference  distribution ( equations (\ref{eq:PURW}) or (\ref{eq:PURWstar}) ) and the factor $\frac{1}{2}$ normalizes the distance. 

As discussed, it is more expensive for the attacker to create txs in the PC with higher numbers of approvers. Hence, it is useful to put increased weight on txs with higher numbers of approvers. However, due to the low probability of having a high number of approvers, the difference between the actual value $P(n)$ and $P_{ref}(n)$ can be relatively large for higher values of $n$. To reward having a high number of approvers (instead of penalizing it because of the high variance), we measure if a tx has more than a given number of approvers, and we can therefore also employ the distance
\begin{equation}\label{eq:d_Q}
    d_{Q}=\frac{1}{2}\sum_{n=0}^\infty |Q(n)-Q_{ref}(n)|
\end{equation}
where 
\begin{equation}
    Q(n)=\sum_{m=n}^\infty P(m)
\end{equation}
is the cumulative probability distribution and $Q_{ref}$ is the reference cumulative probability distribution. As will be shown in the following figures, employing $d_Q$ instead of $d_P$ can lead to a better analysis, whilst adding little or no computational overhead for calculating the more complicated metric $d_Q$.

We can employ these metrics in the following way: if the measured distance exceeds a critical value $\eta$, the detection method rejects this path for the RW. More formally, we say a sample $P$ fails an $\eta$-confidence level test if $d_P > \eta$ (or equivalently, the sample $Q$ fails if $d_Q>\eta$). The exact procedure after rejecting a path can be manifold. For example, the node may simply restart the RW from the last point where no suspicious behavior was noted. Or the node may also switch into some kind of safe mode, where the RW is calculated with a higher $\alpha$ value \cite{ferraro}.

Generally, the measurement of the distance is made on a limited sample size and the distances $d_{P}$ and $d_{Q}$ can, therefore, only take certain discrete values. To determine how likely it is that a measurement is treated as a false-positive, i.e. the sample is part of the main Tangle but has been flagged to be in a PC, we study how likely it is that a certain distance occurs in a random sample. Due to the discreteness of the above distance metrics, we investigate the cumulative probability rather than the actual distribution. Note that for small sample sizes $S$, this discreteness can be observed through the clearly visible discontinuity of the derivative of the curves, see e.g. Fig. \ref{fig:distRWs_CDF} for $S=10$. Intuitively, the cumulative probability also provides a visual representation as to how many samples would be below a given distance, and how many false-positives we should expect, when setting $\eta$.

\subsection{Random walk detection}\label{sec:detection:RW}

In this method of creating a sample of approver numbers, we collect and update the sample set of txs while the RW traverses the Tangle. This method is numerically very inexpensive, since the information for the sample set is readily available, as it is already considered when performing the RW. Hence the node may also decide to test samples of different lengths in parallel, in order to distinguish between very local distributions (small sample size, e.g. $S=10$) and more global distributions (e.g. $S>50)$.

The pseudo code described in Algorithm \ref{alg:TipSelection} provides an example high-level overview of how the PC detection mechanism could be implemented in a tip selection algorithm. The additional algorithm \ref{alg:PCdetection} is added, which updates and evaluates for every RW step a sample of approver numbers. If the calculated distance $d$ is larger than a given threshold, the algorithm \ref{alg:PCdetection} returns a flag and the tip selection method restarts, in this case, from the initial starting point of the RW. Furthermore, the RW is changed from the standard RW (with a $standard\_step$) into a safe mode ($safe\_step$). This safe mode could, for example, be a RW with an increased $\alpha$ value, as described in \cite{ferraro}.

\SetKwFunction{FdetectPC}{\textbf{function} detectPC}
\SetKwInOut{Output}{output}
\SetKwInOut{Input}{input}

\begin{algorithm}[t]
\DontPrintSemicolon
\Input{The starting point of the RW defined as initial\_tx}
\Output{The selected tip} 
tx $\leftarrow$ initial\_tx\;
mode $\leftarrow$ standard\;

\While{tx is \textbf{NOT} a tip}{
    \If{mode = safe}{
    tx $\leftarrow$ safe\_step(tx)\;
    }
    \Else{
        tx $\leftarrow$ standard\_step(tx)\;
        \If{Algorithm~2 = TRUE}{
             mode $\leftarrow$  safe\;
             tx $\leftarrow$ initial\_tx\;
        }
    }
}
\textbf{return} tx
\caption{Integration of detection into tip selection. }
\label{alg:TipSelection}
\end{algorithm}

\begin{algorithm}[t]
\DontPrintSemicolon
\Input{current tx, sample list of approver numbers}
\Output{bool value wheter PC is detected}
Add number of approvers of tx to the sample list\;
\If{Sample size $>$ S }{
Remove the oldest element from the sample list
}
Calculate distance d by applying Eq. (\ref{eq:d_P}) or (\ref{eq:d_Q})\; 
pc\_detected $\leftarrow$ (d $>$ threshold)\;
\textbf{return} pc\_detected
\caption{Sample Management}
\label{alg:PCdetection}
\end{algorithm}

Fig. \ref{fig:distRWs_CDF} shows the cumulative probability for the distance to be above a certain value, where the sample size of considered txs has been fixed to $S$ and $P_{ref}$ is given by (\ref{eq:PURWstar}). It can be seen that for small values of $S$ the discrete nature of the distance is clearly visible, while for larger $S$ it becomes increasingly difficult to observe this phenomenon. Note, that the maximum available number of samples from an RW depends on the depth from which the RW is started. Furthermore, $S$ should be selected such that local changes (i.e. the start of a PC) can be detected, which puts further restrictions on the upper limit of $S$.

\begin{figure}[t]
\begin{minipage}{0.45\textwidth}
    \hspace{-1.0cm}
    \includegraphics[width=1.3\textwidth,trim=0 0 0 0,clip=true]{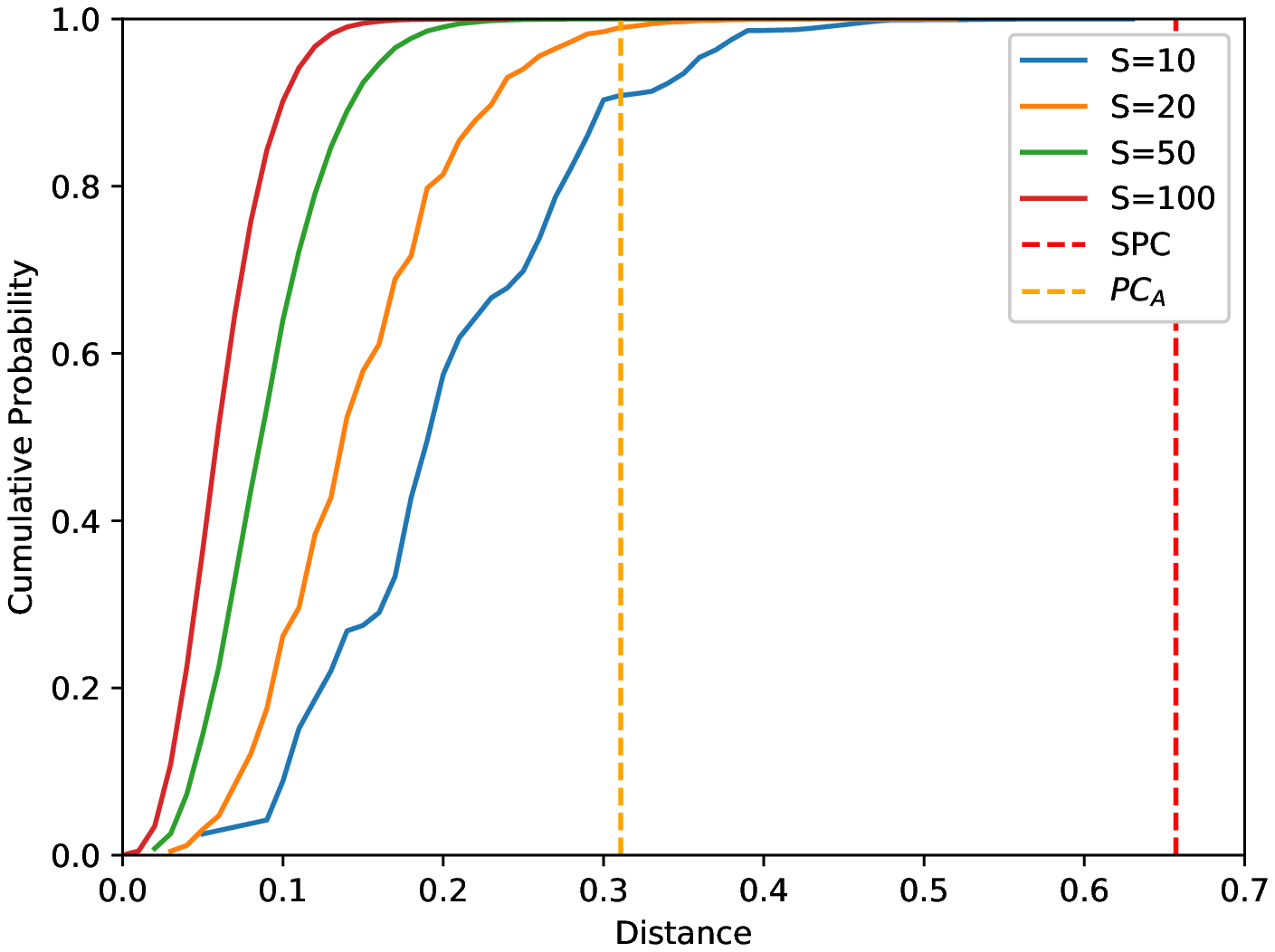}
    a) Distance metric $d_P$
\end{minipage}\hfill
\begin{minipage}{.45\textwidth}
    \hspace{-0.3cm}
    \includegraphics[width=1.3\textwidth,trim=0 0 0 0,clip=true]{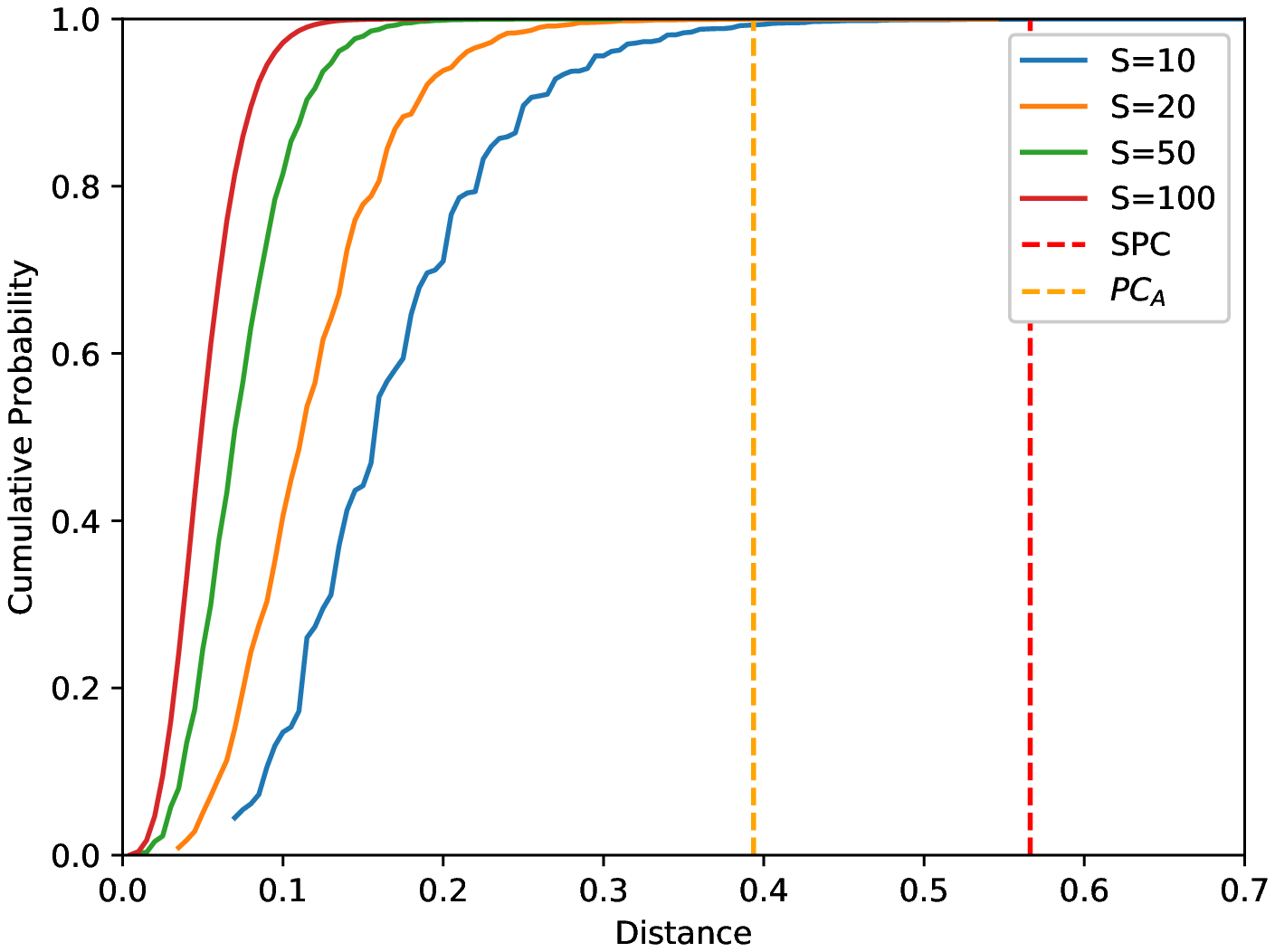}
    b) Distance metric $d_Q$
\end{minipage}\hfill
\vspace{-.2cm}
\caption{Cumulative probability for the measured distance  to be below a given value for several sample sizes $S$. $\lambda=100$, URW. }
\label{fig:distRWs_CDF}
\end{figure}

Let us assume the adversary attaches a 1-pinned PC to the Tangle, and that with probability $p_R$ a malicious tx directly approves the PC root, as well as the previous malicious tx that was attached in the same manner. We denote the part of the PC that is constructed in this way as the \textit{main PC}. The remaining malicious txs are attached in whatever way seems most suitable.
Then, the most efficient way the adversary can build the PC is with $p_R=1$, see Fig. \ref{fig:PC_types}a), where all malicious txs are added as direct approvers to the PC root and the rate at which direct approvers are added to it is given by $r=\mu$.

As can be seen from Fig. \ref{fig:distRWs_CDF}, an $SPC$ is easily detectable, since even for a small number of samples $S$, the measured distance in the PC is larger than the measurements in the main Tangle. By decreasing $p_R$ and attaching to txs along the main PC (1-pinned $PC_1$, see Fig. \ref{fig:PC_types}b) ), the adversary may render the PC undetectable. However, by doing so he weakens the attack since fewer edges are attached to the PC root. Note also that in the $PC_1$ in Fig. \ref{fig:PC_types}b), txs that are not part of the main PC are left behind and due to their low cumulative weight, are therefore unlikely to be selected by the RW as tips.

The selected sample size $S$ should be sufficiently large, since otherwise the detection method may show little success (or return too many false-positives, if the value of $d_C$ is set too low). For example, denote the following PC as $PC_A$: an adversary employs the strategy $PC_1$, by adding a link to the main PC with the probability $1-0.5P_{URW}^*(2)$. Then
\[d_P=1-P^*(1)-P^*(2)\approx 0.26\] 
and the attachment to the root is reduced to the rate
\[r=(1+0.5P_{URW}^*(2))^{-1}\mu\approx 0.85 \mu\]
If $S$ is selected too low (e.g. $S=10$) this type of PC may remain undetected, if $d_P$ (equation \ref{eq:d_P}) is employed for detection, as can be seen from Fig. \ref{fig:distRWs_CDF}a). However, it still may be detectable when changing the detection metric to $d_Q$ (equation \ref{eq:d_Q}).

The PC would certainly remain undetected, if txs are added to the main PC, such that the encountered distribution is $P_{URW}^*(n)$, i.e. $d=0$. Under this condition, the rate is reduced to
\begin{equation}
    r=(1+\sum_{n=1}(n-1)P_{URW}^*(n))^{-1}\mu
\end{equation}
For example, for $\lambda=100$ and $\alpha=0.0001$, $r\approx 0.46\mu$, which requires the attacker to deploy the malicious txs in a much less efficient way,  hence reducing the efficacy of the entire PC attack.

\subsection{Future cone detection}\label{sec:detection:futurecone}

In a pinned PC, the capability of those txs that are pinned to the root, to divert the RW onto the PC relies on their cumulative weight growing as fast as possible. In other words, it is likely that most of the issued malicious txs reference, directly or indirectly, earlier malicious txs. An effective method to consider most of the PC txs is, therefore, to analyze the future cone of the txs that are encountered along an RW, where the future cone of a tx is defined as the set of txs that directly or indirectly approve a given tx.

Note that compared to the method in the previous section, the sample size $S$ can be significantly larger, since the number of samples in the main Tangle would increase exponentially with the distance from the cone's root tx. Until the growth of the future cone reaches the linear phase, where the rate of txs per unit of distance to the root would be constant. However, while the method in Section \ref{sec:detection:RW} is computationally relatively inexpensive, here we must employ a traverse algorithm to efficiently collect a sample of txs, such as Breath-First Search or Depth-First Search \cite{Everitt2018}. It is, therefore, recommended to employ this control mechanism only occasionally or if suspicion is raised.

Due to the larger possible sample size and the way it is sampled, the adversary must attempt to achieve a better agreement of the PC with the reference distribution (\ref{eq:PURW}). For example, the relatively inexpensive PC in \ref{fig:PC_types}b) would exhibit a high distance $d_P$ ($d_Q$), due to the suspiciously large amount of orphans.

The adversary may still create a PC structure that makes the PC less likely to be detected, as Fig. \ref{fig:PC_types}c) indicates. However, since the future cone envelopes the distribution of all txs within a certain distance of the investigated root tx (and in the future cone), the necessary structure becomes significantly more complicated.
More specifically, since the reference distribution itself stems from a structure where all links are employed, this leaves effectively little to no spare links for the adversary to attach to the PC root: due to the fact that the PC is created in secret, the only approvers that the txs can receive before being revealed are the txs from within the PC itself. Hence the average number of approvers in the PC is
\begin{equation*}
\overline{n}_{PC}=\sum_n n P_{PC}(n) \leq \sum_n n P_{ref}(n)= 2
\end{equation*}
where $P_{PC}$ is the probability distribution within the PC and the inequality arises, since the PC creates links to txs within the main Tangle, without receiving any in return. The second equality in the above equation is due to the fact that the total average number of edges per tx in the Tangle must be equal to the number of approvers a tx selects, which is two in the high load regime. Hence, the PC structure that provides optimal resistance to the detection method would require most links to be employed within the PC. The adversary must, therefore, choose a medium between detectability and efficacy of the PC, by deploying links purely within the PC, i.e. he cannot create direct links to the PC root.

In conclusion, the main observation is that, contrary to the approach in Section \ref{sec:detection:RW} where adding a few tx according to the method in \ref{fig:PC_types}b) may be sufficient, in this approach it is more difficult to reproduce the reference distribution. More precisely, in order to imitate the exact distribution, all adversary txs would need to be deployed for mimicking the distribution.

\subsection{Final remarks}

Further improvements can be envisioned for the presented detection tools. For example, we may combine the RW detection in Section \ref{sec:detection:RW} and the future cone detection \ref{sec:detection:futurecone}, since their distributions are different as can be seen by comparing Figs. \ref{fig:alpha0_SEM}a) and \ref{fig:alpha0_SEM}b). Avoiding both detection mechanisms may be noticeably more difficult and their combination could, therefore, lead to an improved success rate. The numerically expensive method in Section \ref{sec:detection:futurecone} may also be more suitable for taking larger samples, while the method in Section \ref{sec:detection:RW} can be easily implemented without much additional numerical effort. Furthermore, the sample size could be continuously varied and alternative sampling mechanisms could be envisioned, which would make it more difficult for the adversary to adopt his strategy.

\section{Conclusion}
We present models to understand and predict a particular aspect of the Tangle structure, more specifically, the likelihood for a transaction to obtain a particular number of direct approvers. We derive solutions for two tip selection mechanisms: firstly, for a Tangle that is built employing a uniform random selection and secondly, if it is constructed employing a random walk. We show that the distributions for the two tip selection methods are linked tightly and also that the distribution depends on whether samples for transactions are taken from the set of all transactions, or only from transactions which are encountered along random walks.

We then employ the models for the approver distributions to detect a specific double-spend attack on the IOTA cryptocurrency, namely the \textit{Parasite Chain} attack. In this type of attack, an adversary attempts to trick a node's tip selection mechanism into approving a double-spend transaction, by building a side Tangle in secret and revealing it at a suitable moment. Since the efficacy of the attack relies on building a side tangle that directly approves a limited set of transactions in the main Tangle, the underlying structure of that Parasite Chain is noticeably different to the main Tangle. By measuring the \textit{distance} to the derived distributions on direct approvers, it is possible to detect certain forms of this side chain. It is shown that the quality of the detection depends on the sample size and we, therefore, propose two different methods of sampling. Firstly, a numerically inexpensive method is proposed, where the node may record the approver statistics along the path of a random walk. However, since the random walk moves relatively fast through the Tangle, a limited amount of transactions can be sampled. In the second approach, the node chooses to sample the future cone of a given transaction. This would ensure that most of the transactions in the Parasite Chain can be measured against the expected distribution.

Through these methods, certain structures of the Parasite Chain may be detectable. This would allow the honest tx issuers to improve their tip selection algorithm to be more safe and may prevent Parasite Chain attacks from being successful. On the other hand, if the adversary chooses to avoid detection by constructing more complex forms that would be more difficult to detect, he has to deploy a significant proportion of transactions in a less effective manner. In either case, the proposed methods present a tool which can reduce the threat imposed by Parasite Chains.


\end{document}